\documentclass{article}

\usepackage{spconf}

\usepackage[T1]{fontenc}
\usepackage{amsmath,amsfonts,amssymb,amsthm,mathtools,accents}
\usepackage{algorithm,algorithmic}

\usepackage{float}

\usepackage[inline]{enumitem}

\usepackage{booktabs}

\usepackage{hyperref}
\hypersetup{
  colorlinks = true,
  allcolors = blue,
}
\usepackage[capitalise,nameinlink]{cleveref}
\usepackage{cite}
\usepackage{doi}

\crefformat{app}{#2Appendix #1#3}
\Crefformat{app}{#2Appendix #1#3}

\usepackage[x11names]{xcolor}

\newcounter{todo}

\makeatletter

\newcommand\listoftodos{\section*{To-do List}\@starttoc{tod}}
\makeatother

\usepackage{tikz}
\usetikzlibrary{matrix,arrows,arrows.meta,
shapes,positioning,
fit,calc,
decorations.pathmorphing,
intersections,
patterns,patterns.meta}

\usepackage{pgfplots}
\pgfplotsset{compat=1.17}
\usepgfplotslibrary{groupplots}
\usepgfplotslibrary{fillbetween}

\newtheorem{theorem}{Theorem}

\newtheorem*{example*}{Example}
\newtheorem*{theorem*}{Theorem}

\theoremstyle{definition}

\theoremstyle{remark}

\newcommand{\vect}[1]{\ensuremath{\mathbf{#1}}} 

\renewcommand{\Vert}[1]{\ensuremath{V_{#1}}}
\newcommand{\Edge}[1]{\ensuremath{E_{#1}}}

\newcommand{\pb}[1]{\ensuremath{#1^{*}}} 
\newcommand{\pf}[1]{\ensuremath{#1_{*}}} 

\newcommand{\gra}[1]{\ensuremath{\mathcal{\uppercase{#1}}}} 
\newcommand{\set}[1]{\ensuremath{\mathcal{\uppercase{#1}}}} 
\newcommand{\bundle}[4]{\ensuremath{\gra{#1}\to\gra{#2}\overset{#4}{\to}\gra{#3}}} 

\newcommand{\ie}{\textit{i.e.}}

\pgfplotsset{colormap={coolwarm}{[1pt]
    rgb(0pt)=(0.2298057, 0.298717966, 0.753683153);
    rgb(1pt)=(0.2389484589019608, 0.3123654946588235, 0.7656759021764705);
    rgb(2pt)=(0.2526625972549019, 0.3328367876470588, 0.7836650259411765);
    rgb(3pt)=(0.26180535615686273, 0.3464843163058824, 0.795657775117647);
    rgb(4pt)=(0.27582712294117645, 0.36671691552941177, 0.812552935372549);
    rgb(5pt)=(0.28527277752941177, 0.38012942263529415, 0.8234685512470589);
    rgb(6pt)=(0.2994412594117647, 0.40024818329411765, 0.8398419750588235);
    rgb(7pt)=(0.30906031906666664, 0.41349827226666663, 0.8501276338666667);
    rgb(8pt)=(0.32371841525490197, 0.4331584405490196, 0.864722355372549);
    rgb(9pt)=(0.3383765114431373, 0.45281860883137254, 0.8793170768784313);
    rgb(10pt)=(0.34832334141176474, 0.4657111465098039, 0.8883461629411764);
    rgb(11pt)=(0.3634607953411765, 0.4847836818509804, 0.9010188868941177);
    rgb(12pt)=(0.37355243129411764, 0.4974987054117647, 0.9094673695294118);
    rgb(13pt)=(0.38885187195294113, 0.5162984355764706, 0.9213734830823529);
    rgb(14pt)=(0.39923148431372546, 0.5285284721568628, 0.9284591027843138);
    rgb(15pt)=(0.41480090285490195, 0.5468735270274511, 0.939087532337255);
    rgb(16pt)=(0.42519897019607844, 0.559058179764706, 0.9460614570784314);
    rgb(17pt)=(0.4411227243607843, 0.5765318648470589, 0.9545453433843137);
    rgb(18pt)=(0.4570464785254902, 0.5940055499294118, 0.963029229690196);
    rgb(19pt)=(0.46767809468235294, 0.6055912316235293, 0.9685462810941176);
    rgb(20pt)=(0.48385432959999997, 0.6220498496, 0.9748082026);
    rgb(21pt)=(0.49463848621176465, 0.6330222615843136, 0.9789828169372549);
    rgb(22pt)=(0.5108243242509803, 0.6493966148235294, 0.9850787763764707);
    rgb(23pt)=(0.5216962808313725, 0.6595986063529412, 0.9877360232470589);
    rgb(24pt)=(0.5380042157019607, 0.6749015936470587, 0.9917218935529412);
    rgb(25pt)=(0.5543118699137254, 0.6900970112156862, 0.9955155482352941);
    rgb(26pt)=(0.5651815812235294, 0.6994384449411764, 0.9966350701176471);
    rgb(27pt)=(0.5814861481882353, 0.7134505955294117, 0.9983143529411764);
    rgb(28pt)=(0.5923558594980393, 0.7227920292549019, 0.9994338748235294);
    rgb(29pt)=(0.6085473603411764, 0.7357252298235294, 0.9993538252980392);
    rgb(30pt)=(0.6193179451882354, 0.7441207347647059, 0.9989309188196078);
    rgb(31pt)=(0.6354738224588236, 0.7567139921764706, 0.9982965591019608);
    rgb(32pt)=(0.6461128107647058, 0.7644364965294117, 0.9968684625058823);
    rgb(33pt)=(0.6619678959411764, 0.7754914668823529, 0.9939365253764706);
    rgb(34pt)=(0.677822981117647, 0.786546437235294, 0.9910045882470588);
    rgb(35pt)=(0.6881884831921569, 0.7931783792980391, 0.9880381043568628);
    rgb(36pt)=(0.7035868880862746, 0.8025856365215686, 0.9828471328745098);
    rgb(37pt)=(0.7138524913490196, 0.8088571413372548, 0.9793864852196078);
    rgb(38pt)=(0.7289695795686274, 0.8174641357058824, 0.973187668372549);
    rgb(39pt)=(0.7388259949411764, 0.8225716218235294, 0.9682610638235294);
    rgb(40pt)=(0.753610618, 0.830232851, 0.960871157);
    rgb(41pt)=(0.7633627801019607, 0.8350922218196078, 0.9556576765568627);
    rgb(42pt)=(0.777377532854902, 0.8409212149490196, 0.9461493015921568);
    rgb(43pt)=(0.7913922856078431, 0.8467502080784314, 0.9366409266274509);
    rgb(44pt)=(0.8006008472941177, 0.8503583215607843, 0.9300075603921568);
    rgb(45pt)=(0.8136925818823529, 0.8542818385490196, 0.9184801025098039);
    rgb(46pt)=(0.8224204049411765, 0.8568975165411765, 0.9107951305882354);
    rgb(47pt)=(0.8353447113529412, 0.8605139972941176, 0.8989704099411765);
    rgb(48pt)=(0.8433581741921568, 0.8618196540156863, 0.8900171168901961);
    rgb(49pt)=(0.8553783684509804, 0.8637781390980391, 0.8765871773137255);
    rgb(50pt)=(0.8674276350862745, 0.864376599772549, 0.8626024620196079);
    rgb(51pt)=(0.8755573874313726, 0.860242158862745, 0.8514300660980393);
    rgb(52pt)=(0.8877520159490196, 0.8540404974980391, 0.8346714722156863);
    rgb(53pt)=(0.8958817682941177, 0.8499060565882353, 0.8234990762941177);
    rgb(54pt)=(0.9061541340352941, 0.8420910651764706, 0.8061505930823529);
    rgb(55pt)=(0.9127650614705882, 0.8366818943529412, 0.7945121117647058);
    rgb(56pt)=(0.9226814526235294, 0.8285681381176471, 0.7770543897882353);
    rgb(57pt)=(0.9281160096666666, 0.8221971488627451, 0.765141349254902);
    rgb(58pt)=(0.9357737696666666, 0.8122367012392158, 0.7471564735843139);
    rgb(59pt)=(0.9434315296666667, 0.8022762536156862, 0.7291715979137255);
    rgb(60pt)=(0.9473454036, 0.7946955048, 0.7169905058);
    rgb(61pt)=(0.9527607176705882, 0.7829647976, 0.6986457713058823);
    rgb(62pt)=(0.9563709270509804, 0.7751443261333334, 0.6864159483098039);
    rgb(63pt)=(0.9605811984235294, 0.7625010185254902, 0.6679635471019607);
    rgb(64pt)=(0.9627082783294117, 0.7535573465568628, 0.655601211227451);
    rgb(65pt)=(0.9658988981882353, 0.7401418386039216, 0.6370577074156862);
    rgb(66pt)=(0.9675442976352941, 0.7308497161882352, 0.6246854782352941);
    rgb(67pt)=(0.9685329496823529, 0.7158412919058823, 0.6060967478823529);
    rgb(68pt)=(0.9695216017294117, 0.7008328676235294, 0.5875080175294117);
    rgb(69pt)=(0.9696829796666666, 0.6904839307372549, 0.5751383613647059);
    rgb(70pt)=(0.9684997476666667, 0.673977379772549, 0.5566492560470588);
    rgb(71pt)=(0.9677109263333333, 0.6629730124627451, 0.5443231858352942);
    rgb(72pt)=(0.9660167198392157, 0.6461297415882352, 0.5258903482588235);
    rgb(73pt)=(0.963806056435294, 0.6341884145294118, 0.5137208491529413);
    rgb(74pt)=(0.9604900613294117, 0.6162764239411764, 0.49546660049411767);
    rgb(75pt)=(0.9566532109764706, 0.598033822717647, 0.4773022923529412);
    rgb(76pt)=(0.9530536002470588, 0.5852108672980392, 0.465372634627451);
    rgb(77pt)=(0.9476541841529411, 0.5659764341686274, 0.4474781480392157);
    rgb(78pt)=(0.9440545734235294, 0.5531534787490197, 0.4355484903137255);
    rgb(79pt)=(0.9367796132117647, 0.5327495001098039, 0.41809333948627453);
    rgb(80pt)=(0.9318312966, 0.5190855232, 0.4064796086);
    rgb(81pt)=(0.9244088216823529, 0.49858955783529413, 0.38905901227058826);
    rgb(82pt)=(0.9182816725843137, 0.48417347218039214, 0.37779392507058823);
    rgb(83pt)=(0.908908026654902, 0.46243263716862765, 0.36095039415294133);
    rgb(84pt)=(0.8995343807254902, 0.4406918021568627, 0.34410686323529416);
    rgb(85pt)=(0.8921375427882353, 0.4253887370980392, 0.33328927276078435);
    rgb(86pt)=(0.8808963866470588, 0.4023312782745098, 0.3171151874901961);
    rgb(87pt)=(0.8734022825529412, 0.3869596390588235, 0.3063324639764706);
    rgb(88pt)=(0.8610536002941176, 0.3629157635294118, 0.2906281271764706);
    rgb(89pt)=(0.8523781350078431, 0.34649194649411763, 0.2803464686980392);
    rgb(90pt)=(0.8393649370784314, 0.32185622094117644, 0.26492398098039216);
    rgb(91pt)=(0.8301865219490197, 0.30473276355294115, 0.25489142806666665);
    rgb(92pt)=(0.8155083866078432, 0.2777809871764706, 0.24029356566666665);
    rgb(93pt)=(0.8008302512666666, 0.2508292108, 0.22569570326666666);
    rgb(94pt)=(0.7905615319411765, 0.23139699905882352, 0.21624203829411764);
    rgb(95pt)=(0.7743368501529412, 0.19975926804705882, 0.2025345544352941);
    rgb(96pt)=(0.763520395627451, 0.17866744737254903, 0.1933962318627451);
    rgb(97pt)=(0.7468380122117647, 0.14002101948235293, 0.17999609695686275);
    rgb(98pt)=(0.7350766252941177, 0.10445963105882351, 0.17149230125490195);
    rgb(99pt)=(0.717434544917647, 0.05111754842352939, 0.15873660770196077);
    rgb(100pt)=(0.705673158, 0.01555616, 0.150232812)
  }}

\usepackage{blindtext}

\ninept

\title{Windowed Fourier Analysis for Signal Processing on Graph Bundles}
\name{%
T.~Mitchell~Roddenberry, Santiago~Segarra
\thanks{This work was supported by USA NSF under award CCF-2008555.}
}

\address{Rice University, Dept. of Electrical and Computer Engineering, Houston, TX, USA}

\begin{document}
\maketitle
\begin{abstract}
  We consider the task of representing signals supported on graph bundles, which are generalizations of product graphs that allow for ``twists'' in the product structure.
  Leveraging the localized product structure of a graph bundle, we demonstrate how a suitable partition of unity over the base graph can be used to lift the signal on the graph into a space where a product factorization can be readily applied.
  Motivated by the locality of this procedure, we demonstrate that bases for the signal spaces of the components of the graph bundle can be lifted in the same way, yielding a basis for the signal space of the total graph.
  We demonstrate this construction on synthetic graphs, as well as with an analysis of the energy landscape of conformational manifolds in stereochemistry.
\end{abstract}
\section{Introduction}

In signal processing and machine learning, a key aspect of many methods is the selection of a proper coordinate system with which to represent a dataset.
Preprocessing steps such as PCA, for instance, represent a dataset in coordinates determined by its principal components, with the hypothesis that only a few coordinates will dominate the rest.
Fourier representations of signal sacrifice spatial or temporal locality in order to make each coordinate (frequency) carry information about the entire signal in question -- wavelets interpolate between the spatial and spectral locality of standard and Fourier representations.
In graph signal processing, we are interested in representing and processing graph signals in ways that reflect the underlying graph geometry.

One such problem in graph signal processing arises when processing signals on product graphs, where factoring the vertex set as the Cartesian product of the vertex set of factor graphs yields a multidimensional graph Fourier transform~\cite{Ortiz2018,Varma2018,Stanley2020}, analogous to the multidimensional Fourier transform in Euclidean space having frequency axes corresponding to each coordinate.
In this case, the Fourier modes of the Cartesian product graph are given by the tensor product of the Fourier modes on the factor graphs.
However, these product factorizations do not always hold.
For instance, a graph may only factor as a product of two factors \emph{locally} about each node, but not globally.
When the product factorization only holds locally, the tensor product factorization of the Fourier modes does not necessarily hold, as it is dependent on the product factorization of the graph holding \emph{globally}.
This is most classically exemplified in the M\"{o}bius graph (\cref{fig:mobius}), where the factors of the graph are approximately given by a cycle graph and a path graph, but there is a twist in the global structure that obstructs such a factorization.
To overcome this obstruction for the task of signal representation, this necessitates the use of a \emph{localized} factored coordinate system about each node.

\subsection{Contributions.}
In this work, we develop tools for representing signals on graphs that locally factor as product graphs.
In particular:
\begin{enumerate}
\item We introduce graph bundles as objects describing graphs that locally factor as product graphs.
\item With the local product structure of graph bundles in mind, we illustrate how any bases for the local factors can be used to construct a localized basis for signals on the graph bundle.
\item We illustrate the utility of this on synthetic graph bundles, as well as an application in analyzing energy landscapes in stereochemistry.
\end{enumerate}

\begin{figure}
  \centering
  \resizebox{0.7\linewidth}{!}{\begin{tikzpicture}
  \begin{groupplot}[
    group style={group size=2 by 1,
      group name=myplots,
      horizontal sep=-0.5cm,
    },
    samples=100,
    samples y=0,
    ]
    
    \nextgroupplot[
        view={-25}{45},
        zmin=-0.2, zmax=0.2,
        hide axis,
    ]
    \foreach \yval in {-0.1,0.1} {
      \addplot3
      [domain=0:360]
      (%
      {(1.0+(\yval)*cos(x/2))*cos(x)},
      {(1.0+(\yval)*cos(x/2))*sin(x)},
      {(\yval)*sin(x/2)}
      );
    }

    \foreach \theta in {0,72,144,216,288} {
      \addplot3
      [mark=square*,mark size=4pt,color=blue]
      coordinates
      {%
        ({(1.0+0.1*cos(\theta/2))*cos(\theta)},{(1.0+0.1*cos(\theta/2))*sin(\theta)},{0.1*sin(\theta/2)})
        ({(1.0-0.1*cos(\theta/2))*cos(\theta)},{(1.0-0.1*cos(\theta/2))*sin(\theta)},{-0.1*sin(\theta/2)})
      };
    }

    \nextgroupplot[
        view={-25}{45},
        zmin=-0.2, zmax=0.2,
        hide axis,
    ]
    \foreach \yval in {-0.1,0.1} {
      \addplot3
      [domain=0:360]
      (%
      {cos(x)},
      {sin(x)},
      {\yval}
      );
    }

    \foreach \theta in {0,72,144,216,288} {
      \addplot3
      [mark=*,mark size=4pt,color=red]
      coordinates
      {%
        ({cos(\theta)},{sin(\theta)},{-0.1})
        ({cos(\theta)},{sin(\theta)},{0.1})
      };
    }
    


    
    
  \end{groupplot}

\end{tikzpicture}}
  
  \vspace{0.5cm}
  
  \resizebox{0.7\linewidth}{!}{\begin{tikzpicture}

  \begin{axis}[
    samples=100,
    samples y=0,
    legend pos=south east,
    xlabel={$j$},
    ylabel={$\lambda_j$},
    label style={font=\Large},
    ticklabel style={font=\large},
    legend style={font=\Large},
    width=12cm,
    height=6cm,
    ]
    \addplot[mark=square*, color=blue] table[x=idx, y=lambda_M] {figs/mobius-spectra.dat};
    \addlegendentry{M\"{o}bius}
    
    \addplot[mark=*, color=red] table[x=idx, y=lambda_C] {figs/mobius-spectra.dat};
    \addlegendentry{Cylinder}
    
  \end{axis}

\end{tikzpicture}}
  
  \vspace{-0.25cm}
  
  \caption{A M\"{o}bius graph (top left), the cylinder graph given by the product of the base and the fiber (top right), and the spectra of their Laplacians (bottom).}
  \label{fig:mobius}
\end{figure}
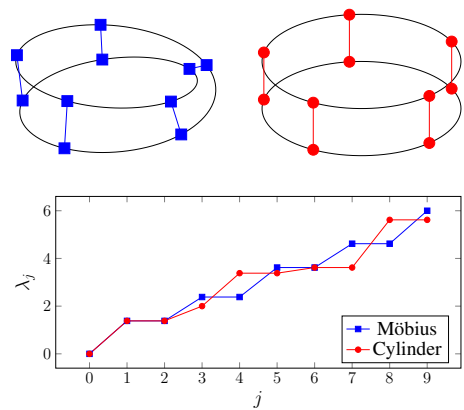

\section{Graph Bundles and Signal Representation}\label{sec:background}

\subsection{Graphs and graph signals.}
A \emph{graph} $\gra{G}$ consists of a finite set of \emph{vertices}, denoted $\Vert{\gra{G}}$, and a set of \emph{edges}, denoted $\Edge{\gra{G}}$, such that every edge is an unordered tuple of vertices.
We find it convenient to treat a graph as a set $\gra{G}=\Vert{\gra{G}}\cup\Edge{\gra{G}}$.
For a given set of vertices $S\subseteq\Vert{\gra{G}}$, the \emph{induced subgraph} on $S$ is the graph whose vertex set is equal to $S$, and whose edge set is comprised of those edges in $\gra{G}$ that span elements of $S$, denoted $\gra{G}[S]$.
For a vertex $v\in\Vert{\gra{G}}$, the \emph{neighborhood graph} about $v$ is the graph $\gra{N}_v$ such that $\Vert{\gra{N}_v}$ consists of $v$ and all nodes $u$ such that $(u,v)\in\Edge{\gra{G}}$ and $\Edge{\gra{N}_v}$ consists of all pairs $(u,v)$.
That is, the neighborhood graph is a star graph centered at $v$.

For graph $\gra{G},\gra{H}$, a \emph{graph map} is a function $\phi:\gra{H}\to\gra{G}$
such that $\phi(\Vert{\gra{H}})\subseteq\Vert{\gra{G}}$,
and for any $(i,j)\in\Edge{\gra{H}}$,
it holds that $\phi((i, j))\in\{\phi(i),\phi(j),(\phi(i), \phi(j))\}$.
We will use the term ``map'' to mean graph map, in general.
For graphs $\gra{G},\gra{H}$ and an injective map $\phi:\gra{H}\to\gra{G}$,
we say that $\phi$ is a \emph{local isomorphism} if for any pair $i,j\in\Vert{\gra{H}}$,
it holds that $(\phi(i),\phi(j))\in\Edge{\gra{G}}$ if and only if $(i,j)\in\Edge{\gra{H}}$.
For two graphs $\gra{B},\gra{F}$, the \emph{(Cartesian) product graph}
is the graph $\gra{G}=\gra{B}\square\gra{F}$
whose vertex set is the Cartesian product $\Vert{\gra{B}}\times\Vert{\gra{F}}$,
and whose edge set is such that $((i_1,i_2),(j_1,j_2))\in\Edge{\gra{G}}$
if and only if $i_1=j_1$ and $(i_2,j_2)\in\Edge{\gra{F}}$,
or $(i_1,j_1)\in\Edge{\gra{B}}$ and $i_2=j_2$~\cite{Godsil2001}.

For a graph $\gra{G}$, a \emph{graph signal} on $\gra{G}$ is a function $\vect{x}:\Vert{\gra{G}}\to\mathbb{R}$.
We denote the set of all such functions by $\mathbb{X}(\gra{G})$,
and endow it with the usual Hilbert space structure via identification with $\mathbb{R}^{\Vert{\gra{G}}}$.
Moreover, we find it useful to define a product of graph signals $(\cdot):\mathbb{X}(\gra{G})\times\mathbb{X}(\gra{G})\to\mathbb{X}(\gra{G})$,
which is evaluated by taking the pointwise product at each vertex of the graph.
For graph signals taking strictly nonnegative values, this yields a well-defined notion of the square-root of a graph signal.

When graphs are related to each other via maps, we can use those maps to relate their signal spaces as well.
Let $\gra{G},\gra{H}$ be graphs, with a map $\pi:\gra{G}\to\gra{H}$ and an injective map $\phi:\gra{H}\to\gra{G}$.
For a signal $\vect{x}\in\mathbb{X}(\gra{H})$, the \emph{pullback} of $\vect{x}$ by the map $\pi$ is a signal on $\gra{G}$ given by precomposition of $\vect{x}$ with $\pi$, denoted by $\pb{\pi}\vect{x}:=\vect{x}\circ\pi$.
Similarly, the \emph{pushforward} of $\vect{x}$ by the injective map $\phi$ is given by precomposition of $\vect{x}$ with $\phi^{-1}$, denoted by $\pf{\phi}\vect{x}:=\vect{x}\circ\phi^{-1}$.
For those vertices $v\in\Vert{\gra{G}}$ not contained in the image of $\phi$, by convention we say that $(\pf{\phi}\vect{x})(v)=0$.

\subsection{Graph bundles.}
\begin{figure}
  \centering
  \resizebox{\linewidth}{!}{\begin{tikzpicture}
  \begin{groupplot}[
    group style={group size=3 by 1,
      group name=myplots,
      horizontal sep=-0.5cm,
    },
    samples=100,
    samples y=0,
    ]

    \nextgroupplot[
    xmin=0, xmax=1,
    ymin=0, ymax=1,
    hide axis,
    ]
    \addplot[color=red, domain=0.2:0.8] ({x}, {x});
    \addplot[mark=*, mark size=4pt, color=red] coordinates {(0.2, 0.2)};
    \addplot[mark=*, mark size=4pt, color=red] coordinates {(0.8, 0.8)};
    
    \nextgroupplot[
    view={-15}{45},
    zmin=-0.2, zmax=0.2,
    hide axis,
    ]
    \foreach \yval in {-0.1,0.1} {
      \addplot3[color=black,domain=0:360]
      (%
      {(1.0+(\yval)*cos(x/2))*cos(x)},
      {(1.0+(\yval)*cos(x/2))*sin(x)},
      {(\yval)*sin(x/2)}
      );
    }

    \foreach \theta in {0,72,144,216,288} {
      \addplot3[color=black]
      coordinates
      {%
        ({(1.0+0.1*cos(\theta/2))*cos(\theta)},{(1.0+0.1*cos(\theta/2))*sin(\theta)},{0.1*sin(\theta/2)})
        ({(1.0-0.1*cos(\theta/2))*cos(\theta)},{(1.0-0.1*cos(\theta/2))*sin(\theta)},{-0.1*sin(\theta/2)})
      };
    }

    \addplot3[mark=*,mark size=4pt,color=black] coordinates
    { ({(1.0+0.1*cos(0/2))*cos(0)},{(1.0+0.1*cos(0/2))*sin(0)},{0.1*sin(0/2)}) ({(1.0-0.1*cos(0/2))*cos(0)},{(1.0-0.1*cos(0/2))*sin(0)},{-0.1*sin(0/2)}) };
    \addplot3[mark=square,mark size=4pt,color=black] coordinates
    { ({(1.0+0.1*cos(72/2))*cos(72)},{(1.0+0.1*cos(72/2))*sin(72)},{0.1*sin(72/2)}) ({(1.0-0.1*cos(72/2))*cos(72)},{(1.0-0.1*cos(72/2))*sin(72)},{-0.1*sin(72/2)}) };
    \addplot3[mark=triangle*,mark size=5pt,color=black] coordinates
    { ({(1.0+0.1*cos(144/2))*cos(144)},{(1.0+0.1*cos(144/2))*sin(144)},{0.1*sin(144/2)}) ({(1.0-0.1*cos(144/2))*cos(144)},{(1.0-0.1*cos(144/2))*sin(144)},{-0.1*sin(144/2)}) };
    \addplot3[mark=o,mark size=4pt,color=black] coordinates
    { ({(1.0+0.1*cos(216/2))*cos(216)},{(1.0+0.1*cos(216/2))*sin(216)},{0.1*sin(216/2)}) ({(1.0-0.1*cos(216/2))*cos(216)},{(1.0-0.1*cos(216/2))*sin(216)},{-0.1*sin(216/2)}) };
    \addplot3[mark=square*,mark size=4pt,color=black] coordinates
    { ({(1.0+0.1*cos(288/2))*cos(288)},{(1.0+0.1*cos(288/2))*sin(288)},{0.1*sin(288/2)}) ({(1.0-0.1*cos(288/2))*cos(288)},{(1.0-0.1*cos(288/2))*sin(288)},{-0.1*sin(288/2)}) };
    
    \nextgroupplot[hide axis]
    \addplot[color=blue, domain=0:360] ({cos(\x)}, {sin(\x)});
    \addplot[mark=*,mark size=4pt,color=blue] coordinates {({cos(0)},{sin(0)})};
    \addplot[mark=square,mark size=4pt,color=blue] coordinates {({cos(72)},{sin(72)})};
    \addplot[mark=triangle*,mark size=5pt,color=blue] coordinates {({cos(144)},{sin(144)})};
    \addplot[mark=o,mark size=4pt,color=blue] coordinates {({cos(216)},{sin(216)})};
    \addplot[mark=square*,mark size=4pt,color=blue] coordinates {({cos(288)},{sin(288)})};

    


    
    
  \end{groupplot}

\end{tikzpicture}}
  \caption{%
    A graph bundle $\textcolor{red}{\gra{F}}\to\gra{G}\overset{\pi}{\to}\textcolor{blue}{\gra{B}}$.
    The projection map $\pi$ is indicated by the shapes of the vertices in the total and base graphs.
  }
  \label{fig:bundle}
\end{figure}
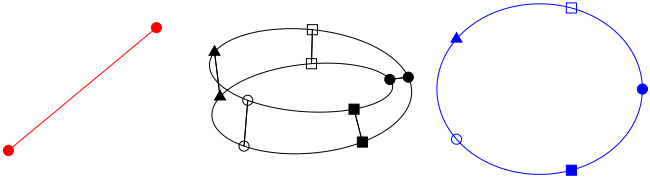
A graph bundle is a simpler construct than a general fiber bundle~\cite{Hatcher2005}.
Let $\gra{F},\gra{G},\gra{B}$ be finite graphs
with a surjective map $\pi:\gra{G}\to\gra{B}$
such that for all $v\in\Vert{\gra{B}}$,
there is a local isomorphism
$\phi_v:\gra{N}_v\square\gra{F}\to\gra{G}$
with the condition that $p_1=\pi\circ\phi$,
where $p_1$ is the projection map onto the first factor.
Under these conditions, we call the object $\bundle{F}{G}{B}{\pi}$ a \emph{graph bundle}.

This is exemplified in \cref{fig:bundle}, once again by the M\"{o}bius graph.
The fiber is a path graph of length $2$, and the base graph is a cycle on $5$ vertices.
The projection from the total graph to the base graph ``squashes'' the fibers each down to a single vertex in the base, thus yielding a cycle graph.

In our example, observe that for any strict subgraph $U\subseteq\gra{B}$, the preimage $\pi^{-1}(U)$ admits a local isomorphism $\phi_U:\pi^{-1}(U)\to U\square\gra{F}$, as desired.
For a graph bundle, there always exists at least one \emph{trivializing cover} of the base graph, which is a set of subgraphs $\set{U}$ of the base graph where
\begin{enumerate*}[label=(\roman*)]
\item $\gra{B}=\cup_{U\in\set{U}}U$, and
\item for each $U\in\set{U}$, there is a local isomorphism
$\phi_U:U\square\gra{F}\to\gra{G}$
with the condition that $p_1=\pi\circ\phi_U$.
\end{enumerate*}
One such trivializing cover is given by the set of all neighborhood graphs in the base, but others can be constructed as well.

A graph bundle naturally generalizes the notion of a (Cartesian) product graph.
For a base graph $\gra{B}$ and fiber $\gra{F}$, let $\gra{G}=\gra{B}\square\gra{F}$ be the Cartesian product graph.
Letting $\pi:\gra{G}\to\gra{B}$ be the standard projection map onto the first coordinate of a product,
one can check that $\bundle{F}{G}{B}{\pi}$ is a graph bundle.
The conditions for a graph bundle imply that for any $v\in\Vert{\gra{B}}$, it holds that $\pi^{-1}(v)\cong\gra{F}$.
We call $\gra{G}$ the \emph{total graph}, $\pi$ the \emph{projection map}, $\gra{B}$ the \emph{base graph}, and $\gra{F}$ the \emph{fiber}.

The advantage of graph bundles is the allowance for ``twists'' in the product structure.
The presence of twists is what distinguishes graph bundles as more general objects than product graphs.
Indeed, a graph bundle with a tree for a base graph factors as a product of the base and the fiber, since there are no loops to twist around over a tree.

\subsection{Spectral graph signal representations.}
Many signal processing and machine learning methods hinge on the choice of a proper representation for the data.
There are a variety of options for graph signal representation put forth in the literature.
For a graph $\gra{G}$, the standard basis for $\mathbb{X}(\gra{G})$ represents a graph signal $\vect{x}$ as a weighted sum of unit impulse functions at each node.
The Fourier basis represents a graph signal $\vect{x}$ as a weighted sum of eigenvectors of a suitably defined graph matrix, such as the Laplacian or any of its normalized variants; this is commonly known as the graph Fourier transform or GFT~\cite{Shuman2013}.
Both the standard and Fourier bases are dictionaries of graph signals that constitute orthonormal bases for the Hilbert space $\mathbb{X}(\gra{G})$.
The standard basis can be though of as being spatially localized and spectrally delocalized, while the Fourier basis is spectrally localized and spatially delocalized.

\section{The Bundle Transform}\label{sec:theory}

\subsection{Local coordinate systems.}
Recent works have considered the processing of signals on \emph{product graphs}~\cite{Ortiz2018,Varma2018,Stanley2020}, which are special cases of graph bundles.
The underlying components of a product graph are used to form a coordinate system, which then allows one to differentiate more interesting spectral features of graph signals using the GFT.
This is analogous to the Fourier transform for images.
Rather than carrying out Fourier analysis isotropically, the Fourier transform of a signal on a grid is understood as a function of two frequencies: one for the horizontal direction, and one for the vertical direction.
For product graphs $\gra{G}=\gra{G}_1\square\gra{G}_2$, the first coordinate can often be thought of as ``space'' and the second as ``time,'' (particularly if $\gra{G}_2$ is a path graph).
Taking the GFT in each coordinate separately allows one to decouple these features, in order to understand the relationship between spatial and temporal Fourier modes.

For a graph bundle, this decoupling is only approximately valid.
A coordinate representation of a graph bundle not only needs a coordinate on the base and a coordinate on the fiber, but also a way to map that coordinate into the total graph via a local isomorphism $\phi_U$.
That is to say, graph bundles indeed have a coordinate system that factors into a base and a fiber, but only when localized to a particular set in the cover of the base graph.

As an example of the GFT's failure to factor over graph bundles, we return to the M\"{o}bius graph in \cref{fig:mobius}, whose base graph is a cycle on five nodes and fiber is a path on two nodes.
If the M\"{o}bius graph could be identified as a product of the base graph and the fiber, then the spectrum of the graph Laplacian would be given by the convolution of the spectra of the base and the fiber~\cite{Merris1994}.
And indeed, applying techniques from~\cite{Roddenberry2022b,Rey2022}, one can check that the $k^{\mathrm{th}}$ moments of the spectra of the product of the base and the fiber and the M\"{o}bius graph are identical up to $k=2$: however, their higher-order moments differ, yielding distinct spectral structures.
This is a result of the graph spectrum being a \emph{global} descriptor of the graph structure, while the property of ``factoring as a product graph'' only holds \emph{locally} for graph bundles.

\subsection{The Bundle Transform.}
\begin{figure}
  \centering
  \includegraphics[width=0.8\linewidth]{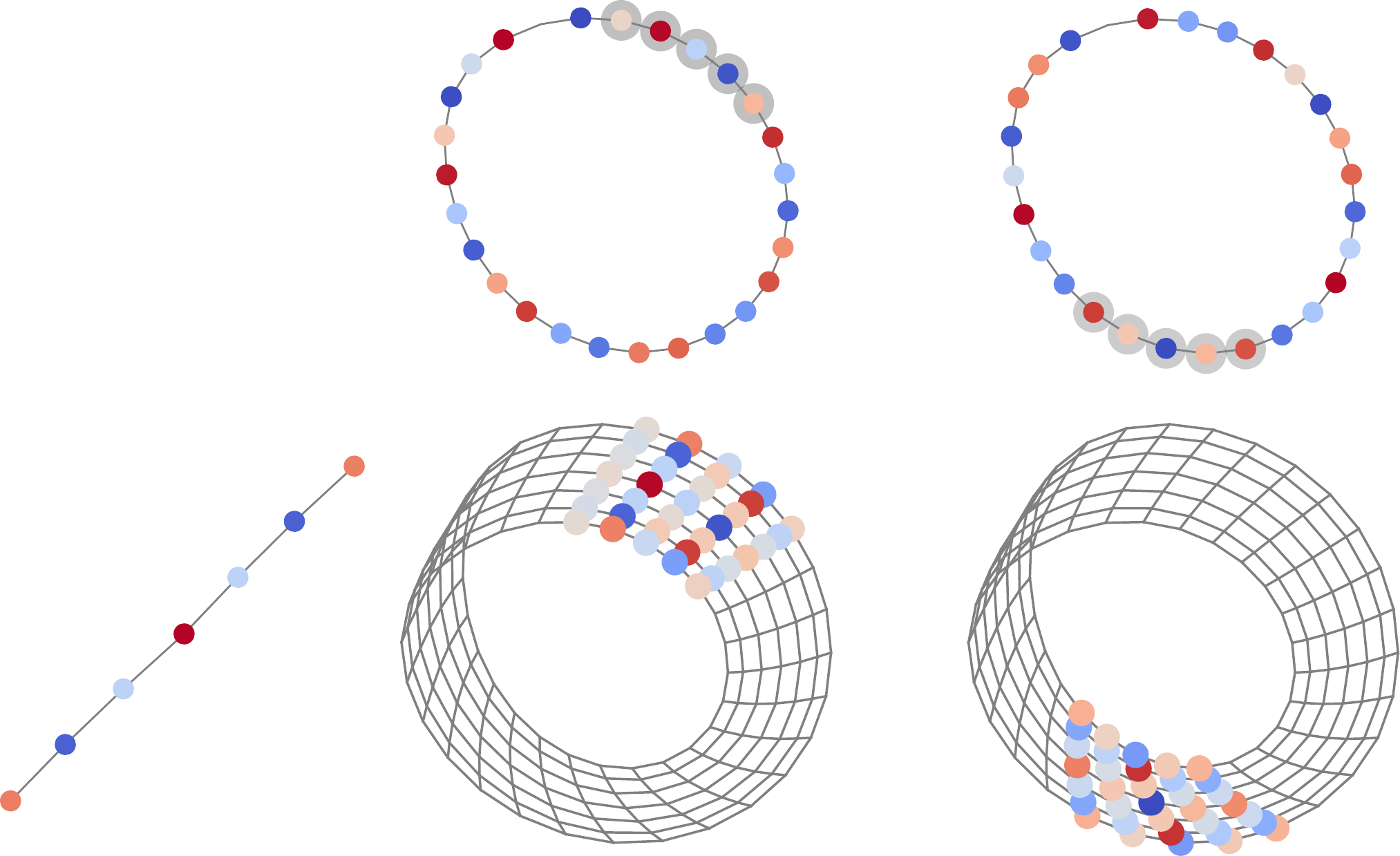}
  \caption{The bundle dictionary on the M\"{o}bius graph.
    A Fourier mode on the fiber is pictured on the left.
    Each column corresponds to a Fourier mode on the base (pictured as the colored signal on the top row) coupled with a covering set and corresponding partition function (highlighted nodes).
  }
  \label{fig:mobius2}
\end{figure}
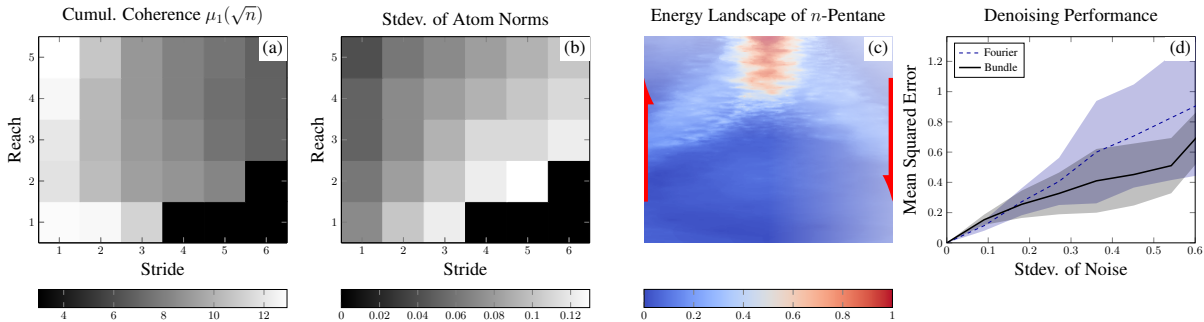
\begin{figure*}
  \centering
  \resizebox{0.9\linewidth}{!}{\begin{tikzpicture}
  \begin{groupplot}[
    group style={group size=4 by 1,
      group name=myplots,
      horizontal sep=1.5cm,
    },
    every axis label/.append style={
      font=\Large,
    },
    title style={
      font=\Large,
    },
    xticklabel style={
      /pgf/number format/fixed,
      /pgf/number format/precision=2,
    },
    ]
    
    \nextgroupplot[
    colormap={WB}{color=(black) color=(white)},
    colorbar,
    view={0}{90},
    grid=none,
    xlabel={Stride},
    ylabel={Reach},
    xlabel style={sloped like x axis},
    title={Cumul. Coherence $\mu_1(\sqrt{n})$},
    xmin=0.5, xmax=6.5,
    ymin=0.5, ymax=5.5,
    point meta min=3,
    colorbar horizontal,
    ]
    \addplot3[surf,shader=flat corner] table[x expr=\thisrow{stride}-0.5, y expr=\thisrow{reach}-0.5, z=mu13] {figs/coherence.dat};

    \nextgroupplot[
    colormap={WB}{color=(black) color=(white)},
    colorbar,
    view={0}{90},
    grid=none,
    xlabel={Stride},
    ylabel={Reach},
    xlabel style={sloped like x axis},
    title={Stdev. of Atom Norms},
    xmin=0.5, xmax=6.5,
    ymin=0.5, ymax=5.5,
    colorbar horizontal,
    ]
    \addplot3[surf,shader=flat corner] table[x expr=\thisrow{stride}-0.5, y expr=\thisrow{reach}-0.5, z=std] {figs/coherence.dat};

    \nextgroupplot[
    hide axis,
    enlargelimits=false,
    title={Energy Landscape of $n$-Pentane},
    colorbar,
    colorbar horizontal,
    colorbar shift/.style={yshift=-1.285cm},
    ]
    \addplot graphics[xmin=0,xmax=1,ymin=0,ymax=1] {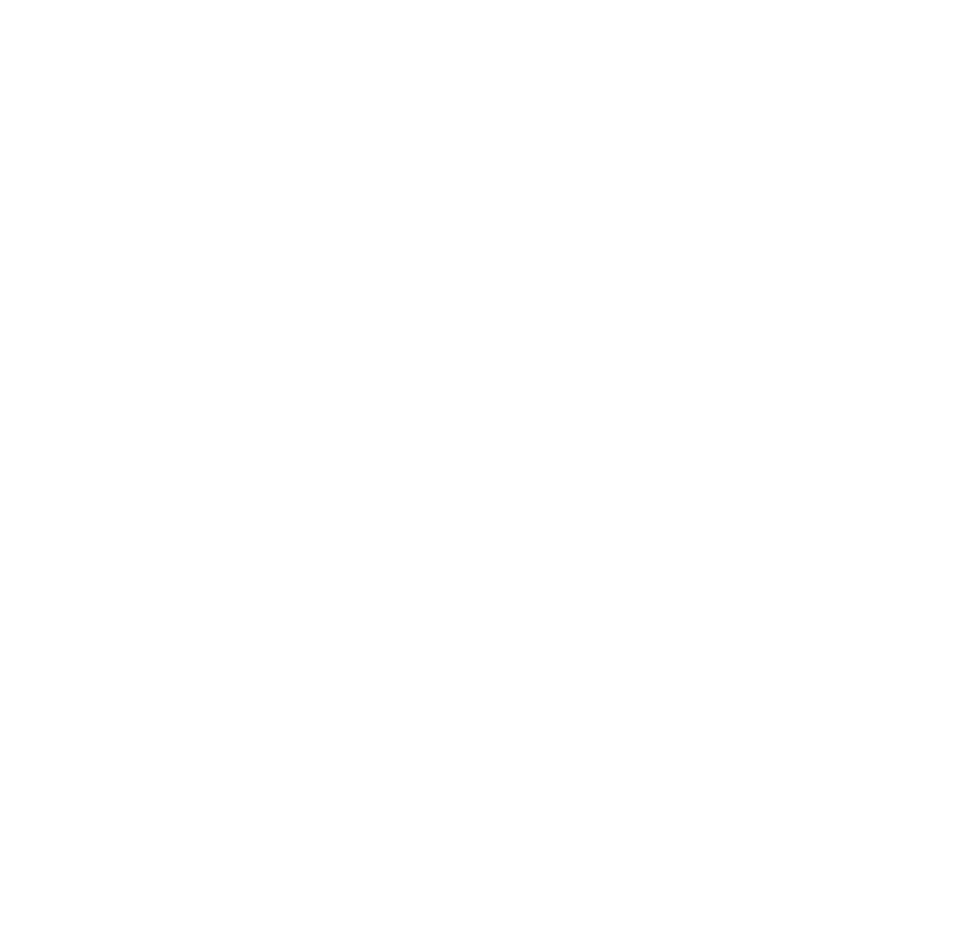};

    \coordinate (a0) at (axis cs:0.0,0.2);
    \coordinate (a1) at (axis cs:0.0,0.8);
    \coordinate (b1) at (axis cs:1.0,0.2);
    \coordinate (b0) at (axis cs:1.0,0.8);

    \draw [-latex,red,line width=6pt] (a0) -- (a1);
    \draw [-latex,red,line width=6pt] (b0) -- (b1);

    \nextgroupplot[
    title={Denoising Performance},
    xlabel={Stdev. of Noise},
    ylabel={Mean Squared Error},
    legend pos={north west},
    xmin=0, ymin=0,
    xmax=0.6,
    ]
    \addplot[name path=A,draw=none,forget plot] table[x=zstd, y=foulb] {figs/denoising.dat};
    \addplot[name path=B,draw=none,forget plot] table[x=zstd, y=fouub] {figs/denoising.dat};
    \addplot[color=blue!60!black,opacity=0.25,forget plot] fill between[of=A and B];

    \addplot[name path=E,draw=none,forget plot] table[x=zstd, y=bunlb] {figs/denoising.dat};
    \addplot[name path=F,draw=none,forget plot] table[x=zstd, y=bunub] {figs/denoising.dat};
    \addplot[black,opacity=0.25,forget plot] fill between[of=E and F];

    \addplot[color=blue!60!black, thick, dashed] table[x=zstd, y=fou] {figs/denoising.dat};
    \addlegendentry{Fourier}
    \addplot[color=black, very thick] table[x=zstd, y=bun] {figs/denoising.dat};
    \addlegendentry{Bundle}
    
  \end{groupplot}

  \foreach \plt/\lab in {c1r1/a,c2r1/b,c3r1/c,c4r1/d} {
    \node[anchor=north east,fill=white] at (myplots \plt.north east) {\Large{(\lab)}};
  }

\end{tikzpicture}}
  \caption{%
    Numerical evaluation of the bundle transform.
    \textbf{(a)} The cumulative coherence of the lifted Fourier bases as a function of the stride and reach of the cover of the base graph. The squares in black are those where the reach is too small for the given stride to yield a cover of the base graph.
    \textbf{(b)} The standard deviation of the atom norms in lifted Fourier bases as a function of the stride and reach of the cover of the base graph.
    \textbf{(c)} The simulated energy landscape of $n$-pentane.
    The identification of the left and right edges is indicated by the red arrows.
    \textbf{(d)} Denoising performance. Lines indicate mean squared error (MSE). Shaded regions indicate the standard deviation of the error over 50 trials.
  }
  \label{fig:exp}
\end{figure*}
We now define the \emph{bundle transform} for signals on graph bundles.
Let $\bundle{F}{G}{B}{\pi}$ be a graph bundle,
with $\set{D}_{\gra{B}}$ an orthonormal basis for $\mathbb{X}(\gra{B})$,
and $\set{D}_{\gra{F}}$ an orthonormal basis for $\mathbb{X}(\gra{F})$.

Let $\set{U}$ be a trivializing cover of $\gra{B}$: for instance, this can be the set of all stars of elements of $\Vert{\gra{B}}$.
Let $\{\rho_U:\in[0,1]^{\Vert{\gra{B}}}:U\in\set{U}\}$ be a partition of unity subordinate to $\set{U}$.
That is, the support of $\rho_U$ is contained in $U$ for each $U\in\set{U}$, and at any $v\in\Vert{\gra{B}}$, the sum $\sum_{U\in\set{U}}\rho_U(v)$ has unit value.

We now construct a set of elements of $\mathbb{X}(\gra{G})$, indexed by $\set{D}_{\gra{B}},\set{D}_{\gra{F}},\set{U}$.
For $\psi_{b}\in\set{D}_{\gra{B}},\psi_{f}\in\set{D}_{\gra{F}},U\in\set{U}$, define a signal $\psi_{b,f}^U\in\mathbb{X}(\gra{G})$ so that
\begin{equation}\label{eq:fiberlet}
  \psi_{b,f}^U = \pf{(\phi_U)}((\sqrt{\rho_U}\cdot\psi_{b}) \otimes \psi_{f}).
\end{equation}
In words, $\psi_{b,f}^U$ is the tensor product of $\psi_{b}$ and $\psi_{f}$ localized to the base set $U$ via the weight function $\rho_U$, then pushed forward to the total graph $\gra{G}$ via the local isomorphism $\phi_U$.
Let $\set{D}_{\gra{G}}$ be the set of all such signals.
As a trivial example, suppose $\set{D}_{\gra{B}}$ and $\set{D}_{\gra{F}}$ are the standard bases consisting of impulse functions, for their respective signal spaces, and the partition of unity is such that each function $\rho_U$ takes value $1$ on a single node and value $0$ elsewhere.
Then, the set $\set{D}_{\gra{G}}$ is merely the standard basis for $\mathbb{X}(\gra{G})$.
For a signal $\vect{x}\in\mathbb{X}(\gra{G})$, the bundle transform of $\vect{x}$ is the collection of inner products $\{\langle\vect{x},\psi\rangle:\psi\in\set{D}_{\gra{G}}\}$.
This construction allows for $\set{D}_{\gra{G}}$ to inherit the representational properties of $\set{D}_{\gra{B}}$ and $\set{D}_{\gra{F}}$, in the following sense:
\begin{theorem}\label{thm:frame}
  Under the assumption that $\set{D}_{\gra{B}}$ and $\set{D}_{\gra{F}}$ are orthonormal bases for the respective signal spaces $\mathbb{X}(\gra{B})$ and $\mathbb{X}(\gra{F})$, the dictionary of atoms $\set{D}_{\gra{G}}$ formed according to \eqref{eq:fiberlet} is a tight frame for $\mathbb{X}(\gra{G})$.
\end{theorem}

\begin{proof}
  Let $\vect{x}\in\mathbb{X}(\gra{G})$ be given arbitrarily.
  We consider the squared $\ell_2$-norm of the transform of $\vect{x}$ by the dictionary $\set{D}_{\gra{G}}$, which can be expressed in the following way:
  \begin{gather*}
    \sum_{\psi_b\in\set{D}_{\gra{B}}}
    \sum_{\psi_f\in\set{D}_{\gra{F}}}
    \sum_{U\in\set{U}}
    \big\langle
    \vect{x},\pf{(\phi_U)}((\sqrt{\rho_U}\cdot\psi_{b}) \otimes \psi_{f})
    \big\rangle^2 = \\
    \sum_{U\in\set{U}}
    \sum_{\psi_b\in\set{D}_{\gra{B}}}
    \sum_{\psi_f\in\set{D}_{\gra{F}}}
    \big\langle
    \pb{\pi}\sqrt{\rho_U}\cdot\vect{x},\pf{(\phi_U)}(\psi_{b} \otimes \psi_{f})
    \big\rangle^2.
  \end{gather*}
  Observe that the map sending $\vect{x}\mapsto\pb{\pi}\sqrt{\rho_U}\cdot\vect{x}$ is an isometry in the following sense:
  \begin{equation*}
    \sum_{U\in\set{U}}\|\pb{\pi}\sqrt{\rho_U}\cdot\vect{x}\|_2^2 = \|\vect{x}\|_2^2,
  \end{equation*}
  owing to the fact that $\{\rho_U\}_{U\in\set{U}}$ is a partition of unity.
  
  For each $U\in\set{U}$, define $\vect{x}_U := \pb{\pi}\sqrt{\rho_U}\cdot\vect{x}$.
  It is then sufficient to show that for each $U\in\set{U}$, the pushforward of the tensor product $\set{D}_{\gra{B}}\otimes\set{D}_{\gra{F}}$ preserves the norm of $\vect{x}_U$.
  For some such $U$ and $\vect{x}_U$, observe that the support of $\vect{x}_U$ is contained in the image of the local isomorphism $\phi_U:U\square\gra{F}\to\gra{G}$.
  Therefore, the pullback of the signal $\vect{x}_U\in\mathbb{X}(\gra{G})$ by $\phi_U$ to $\mathbb{X}(U\square\gra{F})$ is an isometry.
  Clearly, the tensor product of dictionaries $\set{D}_{\gra{B}}\otimes\set{D}_{\gra{F}}$ restricted to $U\square\gra{F}$ forms a tight frame for $\mathbb{X}(U\square\gra{F})$~\cite{Mallat1999}.
  Finally, we have identity of the inner products under the pullback, \ie{},
  \begin{equation*}
    \big\langle
    \vect{x}_U,\pf{(\phi_U)}(\psi_b\otimes\phi_f)
    \big\rangle
    =
    \big\langle
    \pb{(\phi_U)}\vect{x}_U,\psi_b\otimes\psi_f
    \big\rangle.
  \end{equation*}
  With this identity in hand, we can now complete the proof.
  \begin{gather*}
    \sum_{U\in\set{U}}
    \sum_{\psi_b\in\set{D}_{\gra{B}}}
    \sum_{\psi_f\in\set{D}_{\gra{F}}}
    \big\langle
    \pb{\pi}\sqrt{\rho_U}\cdot\vect{x},\pf{(\phi_U)}(\psi_{b} \otimes \psi_{f})
    \big\rangle^2 = \\
    \sum_{U\in\set{U}}
    \sum_{\psi_b\in\set{D}_{\gra{B}}}
    \sum_{\psi_f\in\set{D}_{\gra{F}}}
    \big\langle
    \pb{(\phi_U)}\vect{x}_U,\psi_{b} \otimes \psi_{f}
    \big\rangle^2 = \\
    \sum_{U\in\set{U}}\|\pb{(\phi_U)}\vect{x}_U\|_2^2 =
    \sum_{U\in\set{U}}\|\vect{x}_U\|_2^2 =
    \|\vect{x}\|_2^2,
  \end{gather*}
  as desired.
\end{proof}

The proof is essentially due to~\cite{Donoho1999}, originally for vector bundles over the Grassmanian manifold.
\Cref{thm:frame} indicates that the tightness of the bases for the base and fiber signal spaces are indeed preserved when pushed forward to the total graph.
However, this is at the cost of dispersing the energy of the base atoms across the cover $\set{U}$.
Note that \cref{thm:frame} recovers the known fact that the tensor product of the Fourier modes of the base and the fiber form a basis for the space of signals on their Cartesian product: to see this, let the cover $\set{U}$ of the base graph be the singleton set containing the entirety of $\gra{B}$, so that the corresponding partition of unity is just the constant function on the base graph.

The computation of the atoms $\psi_{b,f}^U$ is illustrated in \cref{fig:mobius2}.
Observe how the support of each atom is restricted to a covering set $U$, which lift to ``patches'' on the total graph via the local isomorphism $\phi_U$.
In the base space, the Fourier modes $\set{D}_{\gra{B}}$ are weighted according to the partition functions $\rho_U$.
This amounts to a windowed graph Fourier basis~\cite{Shuman2012,Shuman2016} on the base graph, which allows for compatibility with the fiber structure.

\section{Experiments}\label{sec:experiments}

\subsection{Locality of the base cover.}
As indicated by \cref{thm:frame}, any pair of orthonormal bases on the base and fiber signal spaces can be locally lifted to the total space via local isomorphisms weighted by a suitable partition of unity.
In this way, the bundle transform guarantees a tight frame for the signal space on the total graph (as opposed to an orthonormal basis, in general, since the frame will be redundant by a factor of the cardinality of the cover $\set{U}$).
We consider here how changing the design of the partition of the unity affects properties of the resulting frame.
Specifically, we evaluate how the cumulative coherence~\cite{Tropp2004} and the standard deviation of the norms of the atoms in $\set{D}_{\gra{G}}$ vary with respect to how far apart the covering sets are situated, and how much of the graph each one covers.

Let the base graph $\gra{B}$ be given by a cycle graph on $27$ nodes, with an extra edge connecting the first and thirteenth nodes in the cycle: this yields a graph consisting of two loops that share an edge.
For a fiber $\gra{F}$ given by a path graph of length $7$, construct a total graph $\gra{G}$ such that the fiber ``twists'' around one of the loops, and doesn't twist around the other.
This yields a cylinder and a M\"{o}bius band that are glued together along the product of a single edge in the base and the fiber.

We construct a family of covers of the base $\gra{B}$ determined by two parameters.
The first is the \emph{stride}, which determines how far apart the centers of the covering sets are located along the base.
The second is the \emph{reach}, which determines the radius of the neighborhood about each center to form the covering sets.
For a cover of the base graph determined by a given stride and reach, we construct a simple partition of unity where each partition function at a particular node in the corresponding covering set is equal to the inverse of the number of sets that cover that node.
As shown in \cref{fig:exp}~(a), the cumulative coherence evaluated at a sparsity level of $\sqrt{n}$, where $n=189$ is the number of nodes in the total graph, the cumulative coherence decreases as the stride increases, with very little influence of the reach.
We see furthermore in \cref{fig:exp}~(b) that standard deviation of the $\ell_2$-norm of the atoms in $\set{D}_{\gra{G}}$ increases with the stride, but decreases with the reach.
A low standard deviation is desirable here, as it indicates the energy is dispersed evenly among the representation of the signal space by the dictionary.
In this case, a design choice needs to be made in the choice of partition function in order to trade off between coherence and variance in the atom norms.

\subsection{Conformation space of pentane.}
In stereochemistry, the energy landscape of molecules is studied with respect to their possible shapes, or \emph{conformations}~\cite{Haile1992}.
One such space of conformations is that of the $n$-pentane molecule, which is a chain of five singly-bonded carbon atoms.
Under a rigid bond model, the bond lengths and bond angles between each consecutive carbon atom is assumed to be fixed,
leaving the only degrees of freedom for the conformation of the molecule to be the torsional angle about the two central carbon-carbon bonds, which we denote by $(\theta_1,\theta_2)$.
By symmetry of the molecule, we can make the identification $(\theta_1,\theta_2)\sim(\theta_2,\theta_1)$.
Then, the conformation space of $n$-pentane can be modeled by the product of two circles modulo this equivalence relation.
The product of two circles yields a torus, and dividing by the equivalence relation yields a M\"{o}bius band.
Each conformation has an associated energy function, determined by the proximity of hydrogen atoms to one another, yielding a function from the M\"{o}bius band to the real numbers.
The stereoisomers of $n$-pentane, then, are those whose conformations are at local minima of this energy function.

We approximate the conformation space by a graph bundle, whose fiber is a path graph of length $6$, and whose base is a cycle graph of length $15$.
The energy landscape is computed using the \texttt{RDKit}~\cite{rdkit} implementation of the universal force field~\cite{Casewit1992}, shown in \cref{fig:exp}~(c).
We consider a denoising task, where the universal frame thresholding method~\cite{Cosentino2020} is applied using a both the Fourier basis of the total graph, and the bundle transform determined by the Fourier modes on the base and fiber graphs, with covering sets of stride $3$ and reach $2$.
The mean squared error of the denoised energy landscape under additive white Gaussian noise using universal frame thresholding is plotted in~\cref{fig:exp}~(d).
Observe that in addition to having the interpretability of locally factoring into the product of the fiber and base coordinates, the lift of the Fourier bases via the bundle transform yields superior performance in denoising over the standard Fourier basis on the total graph.

\section{Conclusion}

Incorporating knowledge of the geometry underlying problems in signal processing and machine learning has been of great interest lately.
A key aspect of this is to represent data in appropriate coordinate systems, allowing for sound interpretation and design of processing methods.
We have considered how ideas from signal processing on products of graphs can be adapted to graph bundles, in which local signal representations are used as ``coordinates'' for structures that only factor as product graphs in a local sense.
This approach preserves tightness of the factor bases, and yields interpretable signal representations, as demonstrated on real and synthetic data.

\newpage

\bibliographystyle{IEEEbib}
\bibliography{ref}

\end{document}